\documentclass[11pt]{article}
\usepackage{amsmath,amsthm,amssymb,cite,enumerate}
\usepackage[a4paper,left=0.5in,right=1in]{geometry}
\usepackage[usenames]{color}
\usepackage{graphicx}
\usepackage{epsfig}
\usepackage{dcolumn}
\usepackage{bm}
\usepackage{authblk} 

\begin{document}

\def \d {{\rm d}}

\def \bm {\mbox{\boldmath{$m$}}}

\def \bF {\mbox{\boldmath{$F$}}}
\def \bV {\mbox{\boldmath{$V$}}}
\def \bff {\mbox{\boldmath{$f$}}}
\def \bT {\mbox{\boldmath{$T$}}}
\def \bk {\mbox{\boldmath{$k$}}}
\def \bl {\mbox{\boldmath{$\ell$}}}
\def \bn {\mbox{\boldmath{$n$}}}
\def \bbm {\mbox{\boldmath{$m$}}}
\def \tbbm {\mbox{\boldmath{$\bar m$}}}

\def \T {\bigtriangleup}
\newcommand{\msub}[2]{m^{(#1)}_{#2}}
\newcommand{\msup}[2]{m_{(#1)}^{#2}}

\newcommand{\be}{\begin{equation}}
\newcommand{\ee}{\end{equation}}

\newcommand{\beqn}{\begin{eqnarray}}
\newcommand{\eeqn}{\end{eqnarray}}
\newcommand{\AdS}{anti--de~Sitter }
\newcommand{\AAdS}{\mbox{(anti--)}de~Sitter }
\newcommand{\AAN}{\mbox{(anti--)}Nariai }
\newcommand{\AS}{Aichelburg-Sexl }
\newcommand{\pa}{\partial}
\newcommand{\pp}{{\it pp\,}-}
\newcommand{\ba}{\begin{array}}
\newcommand{\ea}{\end{array}}

\newcommand{\M}[3] {{\stackrel{#1}{M}}_{{#2}{#3}}}
\newcommand{\m}[3] {{\stackrel{\hspace{.3cm}#1}{m}}_{\!{#2}{#3}}\,}

\newcommand{\tr}{\textcolor{red}}
\newcommand{\tb}{\textcolor{blue}}
\newcommand{\tg}{\textcolor{green}}

\def\a{\alpha}
\def\b{\beta}
\def\g{\gamma}
\def\de{\delta}

\def\E{{\cal E}}
\def\B{{\cal B}}
\def\R{{\cal R}}
\def\F{{\cal F}}
\def\L{{\cal L}}

\def\e{e}
\def\bb{b}

\newtheorem{theorem}{Theorem}[section] 
\newtheorem{cor}[theorem]{Corollary} 
\newtheorem{lemma}[theorem]{Lemma} 
\newtheorem{proposition}[theorem]{Proposition}
\newtheorem{definition}[theorem]{Definition}
\newtheorem{remark}[theorem]{Remark}

\title{Lovelock vacua with a recurrent null vector field}

\author[1,2]{Marcello Ortaggio\thanks{ortaggio(at)math(dot)cas(dot)cz}}

\affil[1]{Institute of Mathematics of the Czech Academy of Sciences, \newline \v Zitn\' a 25, 115 67 Prague 1, Czech Republic}
\affil[2]{Instituto de Ciencias F\'{\i}sicas y Matem\'aticas, Universidad Austral de Chile, \newline Edificio Emilio Pugin, cuarto piso, Campus Isla Teja, Valdivia, Chile}

\maketitle

\abstract{Vacuum solutions of Lovelock gravity in the presence of a recurrent null vector field (a subset of Kundt spacetimes) are studied. 
We first discuss the general field equations, which constrain both the base space and the profile functions. While choosing a ``generic'' base space puts stronger constraints on the profile, in special cases there also exist solutions containing arbitrary functions (at least for certain values of the coupling constants). These and other properties (such as the \pp waves subclass and the overlap with VSI, CSI and universal spacetimes) are subsequently analyzed in more detail in lower dimensions $n=5,6$ as well as for particular choices of the base manifold. 
The obtained solutions describe various classes of non-expanding gravitational waves propagating, e.g., in Nariai-like backgrounds $M_2\times\Sigma_{n-2}$. An appendix contains some results about general (i.e., not necessarily Kundt) Lovelock vacua of Riemann type III/N, and of Weyl and traceless-Ricci type III/N. For example, it is pointed out that for theories admitting a triply degenerate maximally symmetric vacuum, all the (reduced) field equations are satisfied identically, giving rise to large classes of exact solutions.}

\vspace{.2cm}
\noindent

%


\section{Introduction}
\label{intro}

\subsection{Background}

An impressive catalog of exact solutions to Einstein's equations has been obtained over the past century \cite{Stephanibook,GriPodbook}. One particularly interesting family of spacetimes is given by the Kundt metrics, defined by the presence of a null vector field with zero twist, shear and expansion \cite{Kundt61}. These solutions may describe, for example, non-expanding gravitational and electromagnetic waves propagating in various backgrounds, such as Minkowski, (anti-)de~Sitter and (anti-)Nariai \cite{Stephanibook,GriPodbook,Kundt61,PodOrt03}. In any dimension, they contain all near-horizon geometries \cite{KunLuc13} and all VSI (vanishing scalar invariants) spacetimes \cite{Coleyetal04vsi}. They are also important in connection with CSI (constant scalar invariants) spacetimes \cite{ColHerPel09a}, with spacetimes characterized by their curvature invariants \cite{ColHerPel09a}, as well as with universal metrics \cite{Guven87,AmaKli89,HorSte90,Coleyetal08,HerPraPra17} (see, e.g., the review \cite{OrtPraPra13rev} for a summary of some of these properties in arbitrary dimension).

The Kundt family includes, in particular, all spacetimes admitting a {\em recurrent null vector field} \cite{Walker50} (which in turn include \pp waves \cite{Brinkmann25}\footnote{By \pp waves we mean spacetimes admitting a covariantly constant null vector field. These originally appeared in the context of Einstein spacetimes that are conformal to Einstein spacetimes \cite{Brinkmann25} and since then have been studied thoroughly in four \cite{Stephanibook} and higher dimensions \cite{OrtPraPra13rev}.}). This invariantly defined subfamily of Kundt metrics has attracted considerable interest, e.g., because it defines metrics with reduced holonomy contained in Sim$(n-2)$ \cite{KerGol61_0,DebCah61,GibPop08} (of some interest in loop quantum gravity, cf., e.g., \cite{Lewand92} and references therein) and because of its role in the context of universal metrics \cite{Coleyetal08,HerPraPra17}. Moreover, it contains all direct product spacetimes of the form $M_2\times\Sigma_{n-2}$ (which intersect static near-horizon geometries \cite{KunLuc13}). Vacuum solutions to Einstein's theory admitting a recurrent null vector field were studied in four dimensions in \cite{Kundt61,KerGol61} ($\Lambda=0$) and in \cite{LerMcL73,Lewand92} ($\Lambda\neq0$). In view of the increasing interest in higher dimensional gravity, more recently this class of Einstein spacetimes has been discussed in $n$~dimensions \cite{GibPop08} (where other applications are also mentioned).  

However, in more than four dimensions, Einstein gravity can be considered a special case of Lovelock gravity. The latter defines the most general class of theories 
whose field equations are expressed as the vanishing of a symmetric, divergence-free, rank-2 tensor constructed from the metric and its first two derivatives \cite{Lovelock71} (cf., e.g., \cite{GarGir08,Charmousis09,PadKot13} for recent reviews). 
The Lovelock action consists of a finite sum of terms of various order in the curvature tensor with arbitrary coupling constants. In particular, truncating the sum at the linear order gives rise to the Einstein-Hilbert action with a cosmological constant, while the quadratic term corresponds to the Gauss-Bonnet invariant. Such type of interactions is thus also interesting in the low-energy limit of string theory \cite{Zwiebach85}.

It is therefore natural to investigate how known exact solutions of Einstein's theory are modified in the presence of higher order curvature terms in the field equations. Various results have already been obtained, especially with regard to black hole spacetimes (a number of references can be found in \cite{GarGir08,Charmousis09,PadKot13}). 
The purpose of the present paper is to investigate $n$-dimensional vacuum solutions of Lovelock gravity which possess a recurrent null vector field, which appears to be a less explored area. Differences with respect to (and a few similarities with) Einstein's theory will be pointed out, and a few examples presented. 

In the rest of this section we recall the general form of Lovelock's vacuum equations and fix the notation. In section~\ref{sec_recurrent} we summarize the main properties of spacetimes admitting a recurrent null vector field and write down the corresponding curvature components, which will be employed in the following. The field equations for the general line-element are worked out and discussed in section~\ref{sec_equations}. Subsequent sections analyze those in more detail, and present some examples, in special cases such as lower dimensions $n=5,6$ (where only the Gauss-Bonnet term survives; section~\ref{sec_lower}), \pp waves (section~\ref{sec_pp}), for metrics with a base space of constant curvature (either zero or non-zero; section~\ref{sec_const_curv}) or with a base space given by the direct product of two spaces of constant curvature (section~\ref{sec_products}). In the appendix, a few results about general Lovelock vacua of Riemann type III/N or of Weyl and traceless-Ricci type III/N are obtained  (without assuming the presence of a recurrent null vector field).

\subsection{Field equations}

The Lovelock Lagrangian density 
\be
 {\cal L}=\sqrt{-g}\sum_{k=0}^{[(n-1)/2]}c_k{\cal L}^{(k)} , \qquad\qquad {\cal L}^{(k)}=\frac{1}{2^k}\delta_{a_1b_1\ldots a_kb_k}^{c_1d_1\ldots c_kd_k}R_{c_1d_1}^{a_1b_1}\ldots R_{c_kd_k}^{a_kb_k} ,
 \label{Lagr}
\ee
gives rise, in vacuum, to the field equations \cite{Lovelock71}
\be
 G^a_c\equiv\sum_{k=0}^{[(n-1)/2]}c_k G^{a(k)}_{c}=0, \qquad\qquad G^{a(k)}_{c}=-\dfrac{1}{2^{k+1}}\delta^{aa_1b_1\ldots a_kb_k}_{cc_1d_1\ldots c_kd_k}R^{c_1d_1}_{a_1b_1}\ldots R^{c_kd_k}_{a_kb_k} ,
 \label{fieldeqns}
\ee
where $\delta^{a_1\ldots a_p}_{c_1\ldots c_p}=p!\delta^{a_1}_{[c_1}\ldots\delta^{a_p}_{c_p]}$ and $c_k$ are coupling constants. Special choices of the latter correspond, e.g., to Einstein's ($c_k=0$ for $k>1$) or Gauss-Bonnet's theory ($c_k=0$ for $k>2$), including a possible cosmological constant. Indeed, ${\cal L}^{(0)}=1$ corresponds to a cosmological term, while ${\cal L}^{(1)}=R$ is the standard Einstein-Hilbert term (correspondingly, $G^{a(0)}_{c}=-\frac{1}{2}\delta^a_c$, while $G^{a(1)}_{c}=R^a_c-\frac{1}{2}Rg^a_c$ is the Einstein tensor).
The tensors $G^{a(k)}_{c}$ satisfy the Bianchi-like identities $G^{a(k)}_{c;a}=0$ and their traces give $(2k-n){\cal L}^{(k)}=2G^{a(k)}_a$.
 The upper bound in the above summations is due to the fact that ${\cal L}^{(k)}$ does not contribute to the field equations when $2k=n$, and ${\cal L}^{(k)}=0$ for $2k>n$ (i.e., $G^{a(k)}_{c}=0$ for $2k\ge n$). In particular, the case $n=4$ reduces to standard General Relativity, which needs not be discussed again here -- hence we will assume $n\ge 5$ from now on. 

Let us finally recall the useful identity 
\be
 \delta^{a_1\ldots a_s c_{s+1}\ldots c_{k}}_{b_1\,\ldots b_s c_{s+1}\ldots c_{k}}=\frac{(n-s)!}{(n-k)!}\delta^{a_1\ldots a_s}_{b_1\ldots b_s} \qquad (0\le s\le k\le n),
 \label{ident_delta}
\ee
to be employed throughout the paper.

\paragraph{Notation} In $n$ dimensions, we employ a frame of $n$ real vectors $\bm_{(a)} $ which consists of two null vectors $\bl\equiv{\mbox{\boldmath{$m_{(0)}$}}}$,  $\bn\equiv{\mbox{\boldmath{$m_{(1)}$}}}$ and $n-2$ orthonormal spacelike vectors $\bm_{(i)} $, with $a, b\ldots=0,\ldots,n-1$ while $i, j  \ldots=2,\ldots,n-1$, such that $g_{ab}=\ell_a n_b+n_a\ell_b+m_{(i)a}m_{(i)b}$ (cf., e.g., \cite{OrtPraPra13rev} and references therein). For indices $i, j, \ldots$ there is no need to distinguish between subscripts and superscripts. 
Covariant derivatives in the directions of the frame vectors are denoted as
\be
 D \equiv \ell^a \nabla_a, \qquad \T\equiv n^a \nabla_a, \qquad \delta_i \equiv m^{(i)a} \nabla_a . 
 \label{covder}
\ee

\section{Spacetimes with a recurrent null vector field}

\label{sec_recurrent}

In $n$ dimensions, spacetimes that admit a recurrent null vector field $\bl$ \cite{Walker50}, i.e.,
\be
 \ell_{a;b}=\ell_ap_b , \qquad \ell_a \ell^a=0 , 
\label{recur_def}
\ee
coincide with the subclass $\tau_i=0$ ($\Leftrightarrow \ell_{[c}\ell_{a];b}=0$) of the Kundt metrics \cite{Kundt61}. 
Up to a rescaling, $\bl$ can always be chosen such that \cite{DebCah61,KerGol61,Kundt61}\footnote{The results of \cite{DebCah61,KerGol61,Kundt61} were derived in four dimensions, but extend readily to any~$n$, cf., e.g., \cite{GibPop08,OrtPraPra13rev}.} 
\be
 \ell_{a;b}=L\ell_a\ell_b ,
\label{der_l}
\ee
which we will assume hereafter. The special case $L=0$ defines \pp waves \cite{Brinkmann25}, for which $\bl$ is covariantly constant.

With \eqref{der_l}, the Ricci identity readily implies (cf. \cite{OrtPraPra07,GibPop08,Ortaggio09})
\beqn
 & & R_{0i0j}=0 , \\
 & & R_{0i01}=0 , \qquad R_{0ijk}=0 , \\
 & & R_{0i1j}=0 , \qquad R_{01ij}=0 , \label{Riem0=0}
\eeqn
along with 
\be
 R_{0101}=-DL , \qquad  R_{101i}=-\delta_iL . 
\label{R0101}
\ee
{\em The Riemann tensor is thus of type II or more special} \cite{Milsonetal05,OrtPraPra13rev}, aligned with $\bl$. Note also that, by the first of \eqref{Riem0=0}, the (A)dS spacetime does not belong to this class (more generally, spacetimes which differ from (A)dS only by their negative boost-weight (b.w.) Riemann components cannot occur here -- cf. also appendix~\ref{app_RicciIII}).

In adapted coordinates such that $\ell^a\pa_a=\partial_r$ and $\ell_a\d x^a=\d u$, the line-element takes the form \cite{Walker50}
\be
 \d s^2 =2\d u\left[\d r+H(u,r,x)\d u+W_\alpha(u,x)\d x^\alpha\right]+ g_{\alpha\beta}(u,x) \d x^\alpha\d x^\beta , 
 \label{Kundt_gen}
\ee
where $\alpha,\beta=2, \dots, n-1$. Note that the $r$-dependence can only appear in $H$. In these coordinates one has in \eqref{der_l}
\be
  L=H_{,r} .
\ee
so that \eqref{R0101} becomes 
\be
 R_{0101}=-H_{,rr} , \qquad R_{101i}=-\delta_i H_{,r} , 
\label{R0101_2}
\ee
in any frame adapted to $\bl$. If we choose a natural frame by taking
\be
 n_a\d x^a=\d r+H\d u+W_\alpha\d x^\alpha , 
\ee
with the remaining $m_{(i)a}\d x^a=m_{(i)\a}\d x^\a$ defining an o.n. frame for the base space metric $g_{\a\b}$,\footnote{Equivalently, in terms of the contravariant components one has $\bn=\pa_u-H\pa_r$ and $\bm_{(i)} =m_{(i)}^\a(\pa_\a-W_\a\pa_r)$, with $m_{(i)}^\a m_{(j)\a}=\delta_{ij}$.} one further finds (see \cite{ColHerPel06} or use (11p,\cite{OrtPraPra07})\footnote{A missing term in (11p,\cite{OrtPraPra07}) has been pointed out in footnote~7 of \cite{OrtPraPra13rev}.})
\be
  R_{ijkl}=\hat R_{ijkl} , \label{Rijkl} 
\ee
where, from now on, a hat denotes quantities intrinsic to the geometry of $g_{\a\b}$. From \eqref{R0101_2} and \eqref{Rijkl}, it follows that {\em the Riemann type is III (or more special) iff $H_{,rr}=0$ and $g_{\alpha\beta}$ is flat}. It becomes III(a) when, additionally, $\delta_i H_{,r}=0$, in which case the spacetime is necessarily a \pp wave.\footnote{The conditions $H_{,rr}=0=\delta_i H_{,r}$ mean that $DL=0=\delta_i L$, which enables one to set $L=H_{,r}=0$ by an $r$- and $x$-independent boost (this argument was used in section~6 of \cite{HerPraPra14} in the case of Ricci-flat spacetimes).} 

For the frame Riemann components of negative b.w., let us introduce the compact notation 
\be
 {\cal R}_{i}\equiv R_{101i} , \qquad {\cal R}_{ijk}\equiv R_{1ijk} , \qquad {\cal R}_{ij}\equiv R_{1i1j} . \label{Rcomp_1}
\ee
One then has
\beqn 
 & & {\cal R}_{i}=m_{(i)}^\a(-H_{,r\a}+W_\a H_{,rr}) , \label{Ri} \\ 
 & & {\cal R}_{ijk}=m_{(i)}^\a m_{(j)}^\b m_{(k)}^\g R_{u\a\b\g} , \label{Rijk} \\
 & & {\cal R}_{ij}=m_{(i)}^\a m_{(j)}^\b (R_{u\a u\b}+W_\a H_{,r\b}+W_\b H_{,r\a}-W_\a W_\b H_{,rr}) , \label{Rij}
\eeqn
where (cf. \cite{PodSva13}, up to a minor reshuffling) 
\beqn
 R_{u\a\b\gamma}= & & -W_{[\b||\gamma]\a}+\frac{1}{2}(g_{\a\b,u||\gamma}-g_{\a\gamma,u||\b}) , \label{Ruabg} \\
 R_{u\a u\b}=& & -H_{||\a\b}+H_{,r}\left(W_{(\a||\b)}-\frac{1}{2}g_{\a\b,u}\right)+\frac{1}{2}\left(W_{\a,u||\b}+W_{\b,u||\a}\right)-\frac{1}{2}g_{\a\b,uu} \nonumber \\
 & & {}+g^{\gamma\delta}\left(W_{[\a||\gamma]}-\frac{1}{2}g_{\a\gamma,u}\right)\left(W_{[\b||\delta]}-\frac{1}{2}g_{\b\delta,u}\right) . \label{Ruaub} 
\eeqn 
In the above expressions, $W^{\a}\equiv g^{\a\b}W_\b$, and $||$ is the covariant derivative in the base space. Eq.~\eqref{Ri} is equivalent to the second of~\eqref{R0101_2}. For later purposes, let us note that, thanks to \eqref{Ruabg}, ${\cal R}_{ijk}$ does not contain $H$ and thus does not depend on $r$. In addition, in \eqref{Ri}, \eqref{Rij} and \eqref{Ruaub}, $H$ appears always differentiated (once or twice) w.r.t. to $r$, except for the term $H_{||\a\b}$ in $R_{u\a u\b}$.

It follows that the non-zero components of the Ricci tensor are given by (cf. also \cite{GibPop08})
\beqn
 R_{01}= & & -R_{0101}=H_{,rr} , \qquad R_{ij}=\hat R_{ij} , \label{R01} \\
 R_{1i}= & & -{\cal R}_{i}+{\cal R}_{jij}=m_{(i)}^\a\left[-H_{,rr}W_\a+H_{,r\a}-W^{\ \ \ \ \ \ \b}_{[\a||\b]}+\frac{1}{2}{g_{\a\b,u}}^{||\b}-(\ln\sqrt{\hat g})_{,u\a}\right] , \label{R1i} \\
 R_{11}=& & {\cal R}_{ii}=-\Delta H-H_{,rr}W^{\a}W_\a+2H_{,r\a}W^\a+H_{,r}\left[W^{\a}_{\ \ ||\a}-(\ln\sqrt{\hat g})_{,u}\right]+W^{[\a||\b]}W_{[\a||\b]} \nonumber \\
			& &  \qquad\qquad {}+W^{\ \ \ ||\a}_{\a,u}-(\ln\sqrt{\hat g})_{,uu}+\frac{1}{4}g^{\a\b}_{\ \ ,u}g_{\a\b,u} , \label{R11}
\eeqn 
where $\Delta$ is the Laplace operator in the geometry of the base metric $g_{\a\b}$, while 
\be
 R=2H_{,rr}+\hat R .
 \label{R}
\ee
The Ricci type is III when $H_{,rr}=0$ and the base space is Ricci-flat. The Ricci type N occurs if, additionally, $R_{1i}=0$.

To conclude this section, some comments on the coordinate freedom preserving~\eqref{Kundt_gen} will be useful for later purposes. First, an argument of \cite{Brinkmann25} can be easily extended to the present context to show that, by a suitable coordinate transformation, one can always set 
\be
  W_\alpha=0 .
\ee
However, this will in general affect the form of $g_{\a\b}$. When $L=H_{,r}=0$ (i.e., for \pp waves) one can further set $H=0$ \cite{Brinkmann25}. On the other hand, when $W_\alpha$ is a gradient (or, in particular, a constant), it can be simply removed by a suitable shift $r\to r+r_0(u,x)$, which leaves $g_{\a\b}$ unchanged (while $H\to H+r_{0,u}$,). Similarly, when $H=H(u)$, a shift $r\to r+r_0(u)$ (leaving $g_{\a\b}$ and $W_\alpha$ unchanged) can be used to set $H=0$. Finally, a term $H=rh(u)$ can be removed by a rescaling $u\to f(u)$, $r\to r/f'(u)$. These and other  coordinate transformations were discussed in \cite{Brinkmann25, Kundt61,DebCah61,LerMcL73} (see also, e.g., \cite{Coleyetal06}).

\section{Lovelock equations}

\label{sec_equations}

Since the Riemann type of metric~\eqref{Kundt_gen} is II, the only possible non-zero (mixed) components of $G^a_c$ are (ordered by b.w.) $G^0_0$, $G^i_j$, $G^0_i$ and $G^0_1$. Imposing $G^0_0=0$ and $G^i_j=0$ gives, respectively,\footnote{The following summations contain quantities living in the $(n-2)$-dimensional base manifold, which explains the new upper bounds on $k$.}
\beqn
 & & \sum_{k=0}^{[(n-2)/2]} c_k\hat{\cal L}^{(k)}=0 , \label{Lov01} \\
 & & \sum_{k=0}^{[(n-3)/2]} \left[c_k+2c_{k+1}(k+1)H_{,rr}\right]\hat G^{i(k)}_j=0 . \label{Lovij}
\eeqn

Requiring $G^0_i=0$ gives
\be
 \sum_{k=0}^{[(n-3)/2]} c_{k+1}(k+1)\left[(2{\cal R}^j-{\cal R}_s^{js})\hat G^{i(k)}_j-{\cal R}_s^{ij}\hat G^{s(k)}_j-\frac{k}{2^k}{\cal R}_p^{js}\delta^{p i \, j_1 i_2 j_2\ldots i_{k}j_{k}}_{s l_1 s_1 l_2 s_2\ldots l_{k}s_{k}}\hat R_{j j_1}^{l_1 s_1} \hat R_{i_2 j_2}^{l_2 s_2}\ldots \hat R_{i_k j_k}^{l_k s_k} \right]=0 . \label{Lov1i}
\ee

Finally, from $G^0_1=0$ one obtains
\be
 \sum_{k=0}^{[(n-3)/2]} c_{k+1}(k+1)\left[{\cal R}^j_l\hat G^{l(k)}_{j}-\frac{k}{2^{k+1}}\delta^{ijn i_1j_1\ldots i_{k-1}j_{k-1}}_{lsp  \,l_1s_1\ldots l_{k-1}s_{k-1}}{\cal R}_i^{ls}{\cal R}^p_{jn} \hat R_{i_1 j_1}^{l_1 s_1} \ldots \hat R_{i_{k-1} j_{k-1}}^{l_{k-1} s_{k-1}} \right]=0 . \label{Lov11}
\ee 

In the above equations, it is understood that expressions~\eqref{Ri}--\eqref{Rij} hold.

Eqs.~\eqref{Lov01} and \eqref{Lovij} represent restrictions on the metric $g_{\a\b}$ of the base manifold (which must be an Einstein space in Einstein's theory, but this is not the case in general). Further discussion requires to consider separately the two possible cases $H_{,rrr}=0$ and $H_{,rrr}\neq0$ (as can be seen by taking the $r$-derivative of \eqref{Lovij}).  

Before proceeding, let us observe that in the special case of product spaces $M_2\times\Sigma_{n-2}$ (i.e., $H=H(u,r)$, $W_\a=0$, $g_{\a\b,u}=0$ in \eqref{Kundt_gen}), the quantities \eqref{Ri}--\eqref{Rij} are zero, so that the field equations \eqref{Lov1i} and \eqref{Lov11} are satisfied identically. The function $H$ thus enters only in equation~\eqref{Lovij} (via $H_{,rr}$, which gives the Gaussian curvature of $M_2$).

\subsection{Case $H_{,rrr}=0$}

Eq.~\eqref{Lovij} implies that $H_{,rrr}=0$ (unless further restrictions on the base space or on the $c_k$ hold, see section~\ref{subsec_Hspecial} and the examples referred to there), i.e.,
\be
 H(u,r,x)=r^2H^{(2)}(u,x)+rH^{(1)}(u,x)+H^{(0)}(u,x) . 
 \label{deg_Kundt2}
\ee
Hence, eq.~\eqref{Lovij} can be written as
\be
 \sum_{k=0}^{[(n-3)/2]} \left[c_k+4c_{k+1}(k+1)H^{(2)}\right]\hat G^{i(k)}_j=0 . 
 \label{Lovij_1} 
\ee
Once $H^{(2)}$ has been specified, one has to determine a base manifold that solves \eqref{Lov01} and \eqref{Lovij_1} simultaneously.\footnote{In the special case $H^{(2)}=$const, eq.~\eqref{Lovij_1} means that the base space must be itself a solution of a $(n-2)$-dimensional Lovelock theory (with rescaled coefficients $\tilde c_k\equiv c_k+4c_{k+1}(k+1)H^{(2)}$). This happens, for example, in the case of \pp waves, for which $H^{(2)}=0=H^{(1)}$ (cf. section~\ref{sec_pp}). In general, taking the divergence (in the base space) of \eqref{Lovij_1} shows that $H^{(2)}_{,\a}$ must be an eigenvector with eigenvalue zero of the tensor $B^\a_\b\equiv\sum_{k=0}^{[(n-3)/2]} c_{k+1}(k+1)\hat G^{\a(k)}_\b$. In particular, if $B^\a_\b$ has full rank (as happens, e.g., in Einstein's theory, or when the base space has non-zero constant curvature and the theory is generic, cf. section~\ref{subsubsec_const_generic}) then necessarily $H^{(2)}_{,\a}=0$.\label{footn_rank}}
Subsequently, the final step consists in solving \eqref{Lov1i} and \eqref{Lov11}, where the remaining functions $H^{(1)}$, $H^{(0)}$ and $W_\a$ also enter (using \eqref{deg_Kundt2} and \eqref{Ri}--\eqref{Ruaub}, eqs.~\eqref{Lov1i} and \eqref{Lov11} split into various equations of different orders in  $r$ -- however, some of these are identically satisfied by virtue of the others, cf. footnote~8 of \cite{OrtPra16} for a related discussion in Einstein's theory). Examples of solutions in various particular cases will be given in the following sections~\ref{subsec_5_6_generic}, \ref{sec_pp}, \ref{subsubsec_flat_gen}, \ref{subsubsec_const_generic}, \ref{sec_products}.

It is worth noticing that, in the present case, the spacetimes belong to the {\em degenerate} Kundt class \cite{ColHerPel09a,Coleyetal09}. Recall also that all Einstein spacetimes in the Kundt class are necessarily degenerate, so that possible metrics that solve simultaneously both theories (Einstein and Lovelock) must fall within the present branch of solutions.

This branch also contains product spaces $M_2\times\Sigma_{n-2}$ with an $M_2$ of constant Gaussian curvature~$K$, which requires $H=\frac{K}{2}r^2$ (up to a removable term $rH^{(1)}(u)+H^{(0)}(u)$, cf. section~\ref{sec_recurrent}). The metric $g_{\a\b}$ of $\Sigma_{n-2}$ must solve a $(n-2)$-dimensional Lovelock theory (cf. footnote~\ref{footn_rank}) and additionally obey \eqref{Lov01}.

\subsection{Case $H_{,rrr}\neq0$}

\label{subsec_Hspecial}

By considering its $r$-derivative, here \eqref{Lovij} splits into
\be
 \sum_{k=0}^{[(n-3)/2]} c_k\hat G^{i(k)}_j=0  , \qquad \sum_{k=0}^{[(n-3)/2]} c_{k+1}(k+1)\hat G^{i(k)}_j=0 . \label{Lovij_2} 
\ee

This means that the base space must simultaneously solve two different Lovelock's theories in $(n-2)$ dimensions (and additionally obey \eqref{Lov01}) -- one defined by the coefficients $c_k$, and another one by the coefficients $\tilde c_k\equiv (k+1)c_{k+1}$ (however, eqs.~\eqref{Lovij_2} reduce to a single equation in a special theory with $a^{k+1}(k+1)!c_{k+1}=c_0$, where $a\neq0$ is a constant -- a simple example will be given in section~\ref{subsec_specialGB_lower}). 

Assuming a suitable base space is found, one has then to solve \eqref{Lov1i} and \eqref{Lov11}, which (using~\eqref{Lovij_2}) reduce to
\beqn
 & & {\cal R}_p^{js}\sum_{k=0}^{[(n-3)/2]} c_{k+1}(k+1)\frac{k}{2^k}\delta^{p i \, j_1 i_2 j_2\ldots i_{k}j_{k}}_{s l_1 s_1 l_2 s_2\ldots l_{k}s_{k}}\hat R_{j j_1}^{l_1 s_1} \hat R_{i_2 j_2}^{l_2 s_2}\ldots \hat R_{i_k j_k}^{l_k s_k}=0 , \label{Lov1i_spec} \\
 & & {\cal R}_i^{ls}{\cal R}^p_{jn} \sum_{k=0}^{[(n-3)/2]} c_{k+1}(k+1)\frac{k}{2^{k+1}}\delta^{ijn i_1j_1\ldots i_{k-1}j_{k-1}}_{lsp  \,l_1s_1\ldots l_{k-1}s_{k-1}}\hat R_{i_1 j_1}^{l_1 s_1} \ldots \hat R_{i_{k-1} j_{k-1}}^{l_{k-1} s_{k-1}}=0 . \label{Lov11_spec}
\eeqn

The function $H$ does not appear in any of the field equations and thus remains {\em arbitrary}.\footnote{The existence of solutions containing arbitrary functions is a known feature of (certain) Lovelock gravities. To the author's knowledge, this was first noted in \cite{Wheeler86} (for a class of metrics that overlaps only marginally with those considered in the present paper -- cf. section~\ref{subsubsec_const_deg}). The appearence of these ``geometrically free'' solutions \cite{Wheeler86} characterizes, in particular, theories that admit a single (A)dS vacuum \cite{CriTroZan00}. Interestingly, the latter include Chern-Simons ($n$ odd) \cite{Chamseddine89,BanTeiZan94} and Born-Infeld ($n$ even) \cite{BanTeiZan94} gravity.} 
In the special case ${\cal R}_p^{js}=0$ (arising, e.g., when $g_{\a\b,u}=0=W_\a$), eqs.~\eqref{Lov1i_spec} and \eqref{Lov11_spec} are satisfied identically.
Recall that spacetimes with $H_{,rrr}\neq0$ cannot be Einstein and cannot be \pp waves. 
Some examples are mentioned in sections~\ref{subsubsec_special2_lower}, \ref{subsubsec_flat_deg}, \ref{subsubsec_const_deg}, \ref{sec_products} -- however, all of these occur in {\em non-generic} Lovelock theories.

\section{Lower dimensions: $n=5,6$ (Gauss-Bonnet)}

\label{sec_lower}

The b.w.~0 equations~\eqref{Lov01} and \eqref{Lovij} contain the Lovelock scalars and tensors of the $(n-2)$-dimensional base space. When the latter has dimension smaller than five (i.e., when $n=5, 6$), only the terms $\hat{\cal L}^{(0)}$, $\hat{\cal L}^{(1)}$ (and, for $n=6$, $\hat{\cal L}^{(2)}$), $\hat G^{i(0)}_j$ and $\hat G^{i(1)}_j$ are non-zero, so that \eqref{Lov01} and \eqref{Lovij} become considerably simpler. Recall also that, for $n=5, 6$, Lovelock's theory reduces to Gauss-Bonnet's theory.\footnote{However, the field equations of Gauss-Bonnet's theory for $n\ge 7$ would be more complicated, with the Gauss-Bonnet tensor of the base space $\hat G^{i(2)}_j$ entering \eqref{Lovij} -- in contrast to \eqref{Lovij_lower} below. Solutions of Gauss-Bonnet gravity of the form~\eqref{Kundt_gen} in the special case $W_{\a}=0$ have been studied in arbitrary dimensions in \cite{HruPod16}.} It is worth discussing this case in some detail. We can assume $c_2\neq0$, since one is left with Einstein's theory otherwise. 

Eqs. \eqref{Lov01} and \eqref{Lovij} become
\beqn
 & & c_0+c_1\hat R+c_2\hat{\cal L}^{(2)}=0 , \label{Lov01_lower} \\
 & & -\frac{1}{2}(c_0+2c_1H_{,rr})\delta^i_j+(c_1+4c_2H_{,rr})\hat G^{i(1)}_j=0 , \label{Lovij_lower}
\eeqn
with $\hat{\cal L}^{(2)}=0$ identically when $n=5$. Eq.~\eqref{Lovij_lower} shows that $g_{\a\b}$ must be an Einstein metric, unless $c_1+4c_2H_{,rr}=0$, the latter case being possible only when $c_1^2-2c_0c_2=0$. Therefore, the consequences of \eqref{Lov01_lower}, \eqref{Lovij_lower} and of the negative b.w. equations need to be studied separately in the following two cases. (We note that, for $n=6$, the simplifying assumptions $H=H(u,r)$ and $W_{\a}=0=g_{\a\b,u}$ give rise to class-III of \cite{Bogdanosetal09}.)

\subsection{Generic Gauss-Bonnet theory ($c_1^2-2c_0c_2\neq 0$)}

\label{subsec_5_6_generic}

For {\em generic} values of the $c_k$ (in particular, $c_1^2-2c_0c_2\neq 0$), eq.~\eqref{Lovij_lower} implies that $H$ takes the form \eqref{deg_Kundt2}, with
\be
 2H^{(2)}=-\frac{1}{2}\frac{(n-2)c_0+(n-4)c_1\hat R}{(n-2)c_1+2(n-4)c_2\hat R} . 
 \label{H_lower}
\ee
(We observe that the denominator in \eqref{H_lower} cannot vanish here -- by \eqref{Lovij_lower} this would require $c_1^2-2c_0c_2=0$, which is the case studied in section~\ref{subsec_specialGB_lower} below.) Eq.~\eqref{Lovij_lower} also implies that $g_{\a\b}$ is an {\em Einstein} metric (so that $\hat R=\hat R(u)$; one has in addition the condition \eqref{Lov01_lower}, which implies $\hat{\cal L}^{(2)}=\hat{\cal L}^{(2)}(u)$).

Thanks to this, eqs.~\eqref{Lov1i} and \eqref{Lov11} take the form
\beqn
 & & c_1R_{1i}+2c_2\left(\frac{n-4}{n-2}\hat R R_{1i}-\frac{2\hat R}{n-2}{\cal R}^j_{ij}+\hat R^{jl}_{im}{\cal R}^m_{jl}\right)=0 , \label{Lov1i_lower} \\
 & & c_1R_{11}+2c_2\left(\frac{n-4}{n-2}\hat R R_{11}-2{\cal R}^{ij}_i{\cal R}^l_{lj}+{\cal R}^{jl}_i{\cal R}^i_{jl}\right)=0 . \label{Lov11_lower} 
\eeqn
A more explicit but lengthier form of these equations can be obtained using \eqref{Rijk} (with \eqref{Ruabg}), \eqref{R1i}, \eqref{R11} and \eqref{deg_Kundt2}, \eqref{H_lower}, if desired. These can be integrated to determine $H^{(1)}$ and $H^{(0)}$.

For $n=5$, further simplifications occur. First, one can use $c_0+c_1\hat R=0$ (from \eqref{Lov01_lower}) to simplify \eqref{H_lower}. This shows that for $c_0\neq 0$ these spacetimes cannot be Einstein (which would require $2H^{(2)}=\hat R/(n-2)$), and are thus genuine Lovelock solutions (an obvious exception to this would be the case $c_2=0$, which we excluded). Moreover, eq.~\eqref{Lov1i_lower} becomes simpler since the base space is three dimensional (and Einstein), hence of constant curvature, i.e., $\hat R^{jl}_{im}=\frac{1}{3}\hat R\delta^j_{[i}\delta^l_{m]}$ -- this case is thus also contained in the more detailed discussion of section~\ref{subsubsec_const_generic}.

\subsubsection{An example}

\label{subsubsec_5_6_example}

Explicit examples can be easily obtained by making some simplifying assumptions on the metric. For instance, when ${\cal R}^{ij}_i=0$, eqs.~\eqref{Lov1i_lower} and \eqref{Lov11_lower} reduce to $R_{1i}=0=R_{11}$. The simplest possible case when all these conditions are satisfied occurs for $g_{\a\b,u}=0$ (giving $\hat R=$const) and $W_\a=0$, so that the final line-element reduces to (after transforming away a possible term $rH^{(1)}(u)$ in \eqref{Lovij_1})
\be
 \d s^2 =2\d u\d r+2\left[H^{(2)}r^2+H^{(0)}(u,x)\right]\d u^2+ g_{\alpha\beta}(x) \d x^\alpha\d x^\beta , 
 \label{lower_ex}
\ee
with \eqref{H_lower} (so that $H^{(2)}$ is a constant here) and
\be
 \Delta H^{(0)}=0 ,
 \label{lower_ex_eq}
\ee
and where $g_{\alpha\beta}$ can be any Einstein metric obeying \eqref{Lov01_lower}.\footnote{For $n=5$, $g_{\alpha\beta}$ must be of constant curvature, and \eqref{Lov01_lower} is simply a ``normalization'' condition, which can always be satisfied by a constant rescaling of $g_{\alpha\beta}$. For $n=6$, \eqref{Lov01_lower} additionally requires that $\hat{\cal L}^{(2)}$ be a constant -- this is compatible, for example, with $g_{\alpha\beta}$ being a 4-space of constant curvature (as in section~\ref{subsubsec_const_generic}) or a product of two identical 2-spaces of constant Gaussian curvature $\lambda$ (with $c_0+4c_1\lambda+8c_2\lambda^2=0$; cf. also section~\ref{sec_products}), 
the latter choice giving rise to an Einstein spacetime.\label{footn_ex_5D_6D}} Spacetime~\eqref{lower_ex}, \eqref{lower_ex_eq} is generically of Ricci type D and Weyl type II. For $H^{(0)}=0$, it reduces to a direct product spacetime $M_2\times\Sigma_{n-2}$ (of Weyl type D), where $M_2$ possesses constant Gaussian curvature $2H^{(2)}$. For $H^{(0)}\neq0$, it describes non-expanding gravitational waves propagating in such $M_2\times\Sigma_{n-2}$ background (related solutions in arbitrary dimension will be discussed in section~\ref{subsubsec_const_generic}). 

These metrics are not \pp waves, except in the special case $H^{(2)}=0$. For $n=5$ this implies $c_0=0$ and thus $\hat R=0$ (cf. \eqref{Lov01_lower}), which means these are Ricci-flat \pp waves of Weyl type N (and therefore universal spacetime \cite{HerPraPra14}, cf. also section~\ref{subsubsec_flat_gen}). For $n=6$ \pp waves, the Einstein base space must satisfy simultaneously \eqref{Lov01_lower} (so that $\hat{\cal L}^{(2)}$ is a constant) and $2c_0+c_1\hat R=0$ (an example is mentioned in footnote~\ref{footn_L=0}).

\subsection{Special case $c_1^2-2c_0c_2=0$}

\label{subsec_specialGB_lower}

This special case contains, in particular, pure Gauss-Bonnet gravity (when $c_0=0=c_1$). It gives rise to two possibilities.

\subsubsection{$c_1+4c_2H_{,rr}=0$}

This means that $H$ takes the form \eqref{deg_Kundt2}, with
\be
 2H^{(2)}=-\frac{c_1}{4c_2} , 
 \label{H_lower_spec}
\ee
and \eqref{Lovij_lower} is identically satisfied, so that $g_{\a\b}$ has to obey \eqref{Lov01_lower} but need not be Einstein (therefore the full spacetime is also generically non-Einstein). Eqs.~\eqref{Lov1i} and \eqref{Lov11} take the form
\beqn
 & & c_1R_{1i}+2c_2\left(-2G^{j(1)}_iR_{1j}-2\hat R^j_l{\cal R}^l_{ij}+\hat R^{jl}_{im}{\cal R}^m_{jl}\right)=0 , \label{Lov1i_lower_deg} \\
 & & c_1R_{11}+2c_2\left(-2G^{j(1)}_i{\cal R}^i_{j}-2{\cal R}^{ij}_i{\cal R}^l_{lj}+{\cal R}^{jl}_i{\cal R}^i_{jl}\right)=0 . 
\eeqn

(For $n=5$, in \eqref{Lov1i_lower_deg} one can use the identity $\hat R^{ij}_{lm}=4\delta^{[i}_{[l}\hat R^{j]}_{m]}-\hat R\delta^i_{[l}\delta^j_{m]}$.) Metric~\eqref{lower_ex}, \eqref{lower_ex_eq} with \eqref{H_lower_spec} is a solution also here, but $g_{\a\b}$ can now be any metric (not necessarily Einstein) subject to \eqref{Lov01_lower}.

For {\em pure Gauss-Bonnet} gravity the above results specialize to
\beqn
 & & H^{(2)}=0 , \qquad \hat{\cal L}^{(2)}=0 , \\
 & & -2G^{j(1)}_iR_{1j}-2\hat R^j_l{\cal R}^l_{ij}+\hat R^{jl}_{im}{\cal R}^m_{jl}=0 , \\
 & & -2G^{j(1)}_i{\cal R}^i_{j}-2{\cal R}^{ij}_i{\cal R}^l_{lj}+{\cal R}^{jl}_i{\cal R}^i_{jl}=0 . 
\eeqn
Here \eqref{lower_ex} is a solution for any $H^{(0)}$, and for any $g_{\a\b}$ for which $\hat{\cal L}^{(2)}=0$ (hence $g_{\a\b}$ is completely arbitrary for $n=5$).

\subsubsection{$c_1+4c_2H_{,rr}\neq0$}

\label{subsubsec_special2_lower}

In this case \eqref{Lovij_lower} implies that $g_{\a\b}$ must be {\em Einstein} (in addition to obeying \eqref{Lov01_lower}), with
\be
 \hat R=-\frac{n-2}{n-4}\frac{c_1}{2c_2} .
 \label{R_low_spec}
\ee
Thanks to this,  eqs.~\eqref{Lov1i} and \eqref{Lov11} take the simpler form 
\beqn
 & & c_1{\cal R}^j_{ij}+(n-4)c_2\hat R^{jl}_{im}{\cal R}^m_{jl}=0 , \\
 & & -2{\cal R}^{ij}_i{\cal R}^l_{lj}+{\cal R}^{jl}_i{\cal R}^i_{jl}=0 . 
\eeqn
For both $n=5,6$, {\em the function $H(u,r,x)$ is arbitrary} (this is thus a special instance of the case of section~\ref{subsec_Hspecial}), therefore these spacetimes are generically non-Einstein. For $n=5$, eqs.~\eqref{R_low_spec}, \eqref{Lov01_lower} imply that $c_0=0=c_1$ and $\hat R=0$, so that only pure Gauss-Bonnet theory is possible in this branch (see below). For $n=6$, a simple example is given by
 metric~\eqref{Kundt_gen} with $W_{\a}=0$ and with $g_{\a\b}$ being the direct product of two 2-spheres of Gaussian curvature $\lambda=-c_1/4c_2>0$ (which satisfies \eqref{R_low_spec} and \eqref{Lov01_lower} simultaneously, cf. also footnote~\ref{footn_ex_5D_6D}; here $g_{\a\b}$ cannot be of constant curvature, cf. section~\ref{subsubsec_const_deg}). 

For {\em pure Gauss-Bonnet} gravity the above results specialize to
\beqn
 & & \hat R=0 , \qquad \hat{\cal L}^{(2)}=0 , \label{R_L2_low_spec} \\
 & & \hat R^{jl}_{im}{\cal R}^m_{jl}=0 , \label{Lov1i_lower_pure} \\
 & &  -2{\cal R}^{ij}_i{\cal R}^l_{lj}+{\cal R}^{jl}_i{\cal R}^i_{jl}=0 . 
\eeqn
Eqs.~\eqref{R_L2_low_spec} imply that the base space is {\em flat}\footnote{For $n=5$ this is obvious, since $g_{\a\b}$ is of constant curvature and has zero Ricci scalar. For $n=6$, $g_{\a\b}$ is in principle only Ricci-flat, but the condition $\hat{\cal L}^{(2)}=0$ then implies that the Riemann tensor must vanish (since the signature of $g_{\a\b}$ is Euclidean).} and \eqref{Lov1i_lower_pure} is thus identically satisfied -- these solutions are included in those of section~\ref{subsubsec_flat_deg}. For example, any metric~\eqref{Kundt_gen} with $g_{\a\b}=\delta_{\a\b}$ and $W_{\a}=0$ is a solution of pure Gauss-Bonnet gravity (but is not Ricci flat, in general).

\section{\pp waves}

\label{sec_pp}

A special subclass of the spacetimes with a recurrent null vector field $\bl$ consists of \pp waves \cite{Brinkmann25}, for which $\bl$ is covariantly constant and holonomy reduces to (a subgroup of) E$(n-2)$\cite{GibPop08}. These metrics are defined by~\eqref{Kundt_gen} with $H_{,r}=0$, i.e.,
\be
 H=H^{(0)}(u,x) ,
\ee
so that (cf. \eqref{R0101_2}, \eqref{Ri}, \eqref{Rij}) $R_{0101}=0$ and 
\be
  {\cal R}_i=0 , \qquad  {\cal R}_{ij}=m_{(i)}^\a m_{(j)}^\b R_{u\a u\b} . 
	\label{pp_simpl}
\ee
The Riemann type is thus II iff $g_{\a\b}$ is not flat, otherwise it becomes III(a). The latter case (which, by~\eqref{Lov01}, requires $c_0=0$) describes a subset of the VSI metrics, discussed in section~\ref{subsec_flat_base} below. Here we can thus focus on the generic (Riemann type II) case, for which the base manifold is restricted by (cf. \eqref{Lov01}, \eqref{Lovij})
\be
 \sum_{k=0}^{[(n-2)/2]} c_k\hat{\cal L}^{(k)}=0 , \qquad \sum_{k=0}^{[(n-3)/2]} c_k\hat G^{i(k)}_j=0 . \label{Lovijpp} 
\ee
This implies that the spatial metric $g_{\a\b}$ must itself be a solutions of the Lovelock equations in $(n-2)$ dimensions. Recall that in the special case of Einstein's theory, eqs.~\eqref{Lovijpp} mean that $g_{\a\b}$ must be Ricci-flat \cite{Brinkmann25}, but this is not generically the case in Lovelock's theory. 

For example, one may consider as a base space an {\em isotropy-irreducible homogeneous space}, for which necessarily $\hat G^{i(k)}_j=\a_k\delta^i_j$ (with $\a_0=-\frac{1}{2}$) \cite{Bleecker79}. In this case, eqs.~\eqref{Lovijpp} reduce to two algebraic constraints, i.e.,
\be
 \sum_{k=0}^{[(n-2)/2]} c_k\frac{\a_k}{2k-n+2}=0 ,  \qquad \sum_{k=0}^{[(n-3)/2]} c_k\a_k=0 , 
\ee
while \eqref{Lov1i}, \eqref{Lov11} become
\beqn
 & & \sum_{k=0}^{[(n-3)/2]} c_{k+1}(k+1)\left[2\a_kR_{1i}+\frac{k}{2^k}{\cal R}_p^{js}\delta^{p i \, j_1 i_2 j_2\ldots i_{k}j_{k}}_{s l_1 s_1 l_2 s_2\ldots l_{k}s_{k}}\hat R_{j j_1}^{l_1 s_1} \hat R_{i_2 j_2}^{l_2 s_2}\ldots \hat R_{i_k j_k}^{l_k s_k} \right]=0 , \\
 & & \sum_{k=0}^{[(n-3)/2]} c_{k+1}(k+1)\left[\a_kR_{11}-\frac{k}{2^{k+1}}\delta^{ijn i_1j_1\ldots i_{k-1}j_{k-1}}_{lsp  \,l_1s_1\ldots l_{k-1}s_{k-1}}{\cal R}_i^{ls}{\cal R}^p_{jn} \hat R_{i_1 j_1}^{l_1 s_1} \ldots \hat R_{i_{k-1} j_{k-1}}^{l_{k-1} s_{k-1}} \right]=0 , 
\eeqn 
with \eqref{R1i}, \eqref{R11}. 

In general, the explicit integration of the remaining field equations \eqref{Lov1i}, \eqref{Lov11} (with \eqref{pp_simpl}), which constrain $W_{\a}$ and $H^{(0)}(u,x)$, will depend on the chosen base space. Let us only observe here that these simplify considerably if one chooses coordinates such that $W_\a=0=H$ (which is always locally permitted \cite{Brinkmann25}, cf. section~\ref{sec_recurrent}), since \eqref{Ruabg} and \eqref{Ruaub} then reduce to 
\be
 R_{u\a\b\gamma}= \frac{1}{2}(g_{\a\b,u||\gamma}-g_{\a\gamma,u||\b}) , \qquad R_{u\a u\b}=-\frac{1}{2}g_{\a\b,uu}+\frac{1}{4}g^{\gamma\delta}g_{\a\gamma,u}g_{\b\delta,u} \qquad (W_\a=0=H) .
\ee
If one further {\em assumes} that (in the same coordinate system) also $g_{\a\b,u}=0$, then $R_{u\a\b\gamma}=0=R_{u\a u\b}$, the Riemann type is D, and  \eqref{Lov1i}, \eqref{Lov11} are identically satisfied.   Examples of \pp waves (also allowing for $W_\a\neq0$ or $H\neq0$) are mentioned in sections~\ref{subsubsec_5_6_example}, \ref{subsubsec_flat_gen}, \ref{subsubsec_const_generic} (footnote~\ref{footn_L=0}).

\section{Base space of constant curvature}

\label{sec_const_curv}

Another interesting and tractable case occurs when the base space is of constant curvature, either zero (section~\ref{subsec_flat_base}) or non-zero (section~\ref{subsec_const_base}). This corresponds, respectively, to having a Minkowski or a Nariai-like ``background''.

\subsection{Flat base space, and VSI solutions}

\label{subsec_flat_base}

The case when $g_{\a\b}$ is a flat metric is of considerable interest, since it includes all solutions of the form~\eqref{Kundt_gen} which are of Riemann type III (cf. section~\ref{sec_recurrent}) and therefore, in particular, which possess the VSI property \cite{Coleyetal04vsi}.\footnote{Recall, however, that not all VSI spacetimes admit a recurrent null vector field \cite{Coleyetal04vsi,Coleyetal06}.} It also includes the case when $\bl$ is a (recurrent) Kerr-Schild vector field.

If $g_{\a\b}$ is flat, then $\hat R_{ijkl}=0$, so that $\hat G^{(k)i}_{\ \ j}=0=\hat{\cal L}^{(k)}$ for $k\ge 1$. This means that the field equations \eqref{Lov01}--\eqref{Lov11} reduce to
\beqn
 & & c_0=0 , \label{flat1} \\
 & & c_1H_{,rr}=0 , \label{flat2} \\
 & & c_1R_{1i}=0 , \label{flat3} \\
 & & c_1R_{11}+2c_2(-2{\cal R}^{ij}_i{\cal R}^l_{lj}+{\cal R}^{jl}_i{\cal R}^i_{jl})=0 . \label{flat4} 
\eeqn

We observe that here eqs.~\eqref{flat1}--\eqref{flat3} take the same form as in Einstein's theory (with the cosmological constant necessarily vanishing due to $c_0=0$), while \eqref{flat4} contains also the Gauss-Bonnet term, but not terms of higher order.

\subsubsection{Generic Lovelock theory ($c_1\neq0$)}

\label{subsubsec_flat_gen}

When $c_1\neq0$, \eqref{flat2} implies  
\be
 H(u,r,x)=rH^{(1)}(u,x)+H^{(0)}(u,x) , 
 \label{H_lin}
\ee
while \eqref{flat3} gives $R_{1i}=0$.
Without loss of generality, we can choose coordinates such that 
\be
 g_{\a\b}=\delta_{\a\b} .
\label{g_flat}
\ee 

Then, eq.~\eqref{flat3} with \eqref{R1i}, \eqref{H_lin} determines $H^{(1)}$ (as in \cite{Coleyetal06}) via
\be
 H^{(1)}_{,\a}=W^{\ \ \ \ \ \b}_{[\a||\b]} , \label{H1} 
\ee
such that $\Delta H^{(1)}=0$, while \eqref{flat4} with \eqref{R11} and \eqref{Ruabg} gives $H^{(0)}$ as a solution of (in this case, covariant derivatives in the base space reduce to ordinary derivatives) 
\beqn
 & & c_1\left(\Delta H^{(0)}-H^{(1)}W^{\a}_{\ \ ||\a}-2W^{\a}H^{(1)}_{,\a}-W^{[\a||\b]}W_{[\a||\b]}-W^{\ \ \ ||\a}_{\a,u}\right)  \nonumber \\
 & & \qquad =2c_2\left(-2W_{[\a||\b]}^{\ \ \ \ \ \a}W^{[\gamma||\b]}_{\ \ \ \ \ \gamma}+W_{[\a||\b]\gamma}W^{[\a||\b]\gamma}\right) . \label{H0}
\eeqn

Generically, the Weyl type is III and the Ricci type is N, and it can become more special for particular choices of the metric functions \cite{Coleyetal06} (however, non-flat, conformally flat solutions are not possible, since \eqref{flat3} and \eqref{flat4} imply that the full Riemann tensor vanishes if the Weyl tensor does). This implies that the metric belongs to the VSI class. The spacetime is a \pp wave when $H^{(1)}=0$ (so that ${\cal R}_i=0$). As an example, one can take \eqref{Kundt_gen} with \eqref{H_lin}, \eqref{g_flat} and 
\beqn
 & & W_\a=(ax_3^2+bx_4^2)\delta_{\a,2} , \qquad H^{(1)}=(a+b)x_2 , \label{ex_g_flat1} \\
 & & H^{(0)}=\frac{a}{6}(2a+b)x_3^4+\frac{b}{6}(a+2b)x_4^4+2\left(c-ab\frac{c_2}{c_1}\right)x_3^2-2\left(c+ab\frac{c_2}{c_1}\right)x_4^2 , \label{ex_g_flat2}
\eeqn
where $a,b,c$ are arbitrary functions of $u$ (giving rise to a \pp wave iff $a+b=0$).

In the special case of Ricci-flat metrics of Weyl type III, eq.~\eqref{flat4} reduces to $c_2(-2{\cal R}^{ij}_i{\cal R}^l_{lj}+{\cal R}^{jl}_i{\cal R}^i_{jl})=0$ (cf. also Proposition~\ref{prop_RiemIII}). By theorem~1.4 of \cite{HerPraPra14}, this implies (except for special theories with $c_2=0$) that these Ricci-flat, Weyl type III Lovelock solutions not only solve Einstein and Lovelock gravity, but are in fact universal spacetimes (these cannot be \pp waves, unless the Weyl type degenerates to N \cite{MalPra11prd,HerPraPra14}). An example is given by \eqref{ex_g_flat1}, \eqref{ex_g_flat2}, with $ab=0$ (see section~6.2.2 of \cite{HerPraPra14} for different examples).

When the Weyl type is N, necessarily ${\cal  R}^{ij}_l=0$, so that, by \eqref{flat4}, the metric must be Ricci-flat (cf. also Proposition~\ref{prop_RiemN}). Furthermore, it must be a \pp wave (see section~\ref{sec_recurrent}). One can thus set $H^{(1)}=0$, $W_\a$ must take the form given in \cite{Coleyetal06} ($W_\a=0$ being a special case thereof), and the r.h.s. of \eqref{H0} is identically zero. An example is given by \eqref{ex_g_flat1}, \eqref{ex_g_flat2}, with $a=0=b$. Let us observe that Ricci-flat type N \pp waves are another instance of universal spacetimes \cite{HerPraPra14} (see also the earlier results \cite{Guven87,AmaKli89,HorSte90} in the special case $W_\a=0$).

\subsubsection{Lovelock theories with $c_0=0=c_1$}

\label{subsubsec_flat_deg}

This is clearly a very special case (including, in particular, pure Gauss-Bonnet gravity -- cf. section~\ref{subsubsec_special2_lower} in the case $n=5,6$) since the field equations reduce to the single equation
\be
  c_2(-2{\cal R}^{ij}_i{\cal R}_l^{lm}+{\cal R}^{jl}_i{\cal R}^i_{jl})=0 . 
	\label{flat_special}
\ee

The function $H$ does not enter ${\cal R}^{ij}_l$ and is thus arbitrary. The Riemann type is II (or D) as long as $H_{,rr}\neq0$. Without loss of generality, we can choose coordinates such that $g_{\a\b}=\delta_{\a\b}$. Eq.~\eqref{flat_special} thus takes the form \eqref{H0} with $c_1=0$. Any metric~\eqref{Kundt_gen} with $g_{\a\b}=\delta_{\a\b}$ and $W_{\a}$ as in \eqref{ex_g_flat1} with $ab=0$ is clearly a solution. More generally, one can take any metric~\eqref{Kundt_gen} (with $g_{\a\b}=\delta_{\a\b}$) of Riemann type N (a subset of the metrics of \cite{Coleyetal06}), or of Riemann type III and satisfying \eqref{flat_special} (some examples are mentioned in \cite{MalPra11prd}), and modify the function $H$ arbitrarily, thus obtaining a new (in general non-isometric) solution. Such metrics are VSI iff $H_{,rr}=0$.

For the subset of theories for which also $c_2=0$ (non-trivial for $n\ge7$), {\em any} metric~\eqref{Kundt_gen} with a flat $g_{\a\b}$ obviously satisfies the field equations.

\subsection{Base space of non-zero constant curvature}

\label{subsec_const_base}

Here we consider the case when
\be
 \hat R^{ij}_{lp}=\lambda\delta^{ij}_{lp} ,
 \label{Riem_const}
\ee
where $\lambda$ is a constant,\footnote{We note that $\lambda$ cannot depend on the $x$ thanks to the Bianchi identity in the base space, while $\lambda_{,u}\neq0$ is ruled out by \eqref{const1} below.} which implies (using \eqref{Lagr}, \eqref{fieldeqns})
\be
 \hat{\cal L}^{(k)}=\frac{(n-2)!}{(n-2k-2)!}\lambda^k ,  \qquad \hat G^{i(k)}_j=-\frac{\lambda^k}{2}\frac{(n-3)!}{(n-2k-3)!}\delta^i_j , 
\ee
where it is understood that $k\le[(n-2)/2]$ in the first equation and $k\le[(n-3)/2]$ in the second one (the above quantities being zero otherwise). 
Using these and defining 
\be
 P(\lambda)\equiv \sum_{k=0}^{[(n-1)/2]} c_k\frac{\lambda^k}{(n-2k-1)!} , \qquad Q(\lambda)\equiv (n-1)P(\lambda)-2\lambda P'(\lambda) ,
\ee
where $P'$ is the derivative of $P$ w.r.t. $\lambda$, the field equations \eqref{Lov01}--\eqref{Lov11} reduce to
\beqn
 & & Q(\lambda)=0 , \label{const1} \\
 & & (n-2)Q(\lambda)-2\lambda Q'(\lambda)+2H_{,rr}P'(\lambda)=0 , \label{const2} \\
 & & (n-3)P'(\lambda)R_{1i}-2\lambda P''(\lambda){\cal R}_{ij}^j=0 , \label{const3} \\
 & & (n-3)(n-4)P'(\lambda)R_{11}+P''(\lambda)(-2{\cal R}^{ij}_i{\cal R}^l_{lj}+{\cal R}^{jl}_i{\cal R}^i_{jl})=0 . \label{const4} 
\eeqn

The polynomial eq.~\eqref{const1} generically possesses $[(n-2)/2]$ solutions, determining $\lambda$ in terms of the $c_k$. We assume that at least one of such solutions, say $\bar\lambda$, is real (otherwise the Lovelock equations do not admit solutions in the present subclass of spacetimes) and non-zero. Then one has to solve \eqref{const2}--\eqref{const4} for this particular value $\bar\lambda\neq0$.

\subsubsection{Generic Lovelock theory ($P'(\bar\lambda)\neq0$)}

\label{subsubsec_const_generic}

Generically, $P'(\bar\lambda)\neq0$ (cf. section~\ref{subsubsec_const_deg}). With \eqref{const1}, eq.~\eqref{const2} means that 
\be
 H_{,rr}=\bar\lambda\frac{Q'(\bar\lambda)}{P'(\bar\lambda)}\equiv \frac{2\Lambda_0}{n-2} 
 \label{Hrr}
\ee
is a (generically non-zero) constant, which gives\footnote{We have defined the constant $\Lambda_0$ such that it reduces to the usual cosmological constant in the limit of Einstein's gravity, with the normalization $R=2n\Lambda_0/(n-2)$.}
\be
 H(u,r,x)=\frac{\Lambda_0}{n-2}r^2+rH^{(1)}(u,x)+H^{(0)}(u,x) . \label{H_quad}
\ee

Since $g_{\a\b}$ is a metric of constant curvature, with no loss of generality we can choose coordinates such that (cf., e.g., section~4 of \cite{ColHerPel06})
\be
 g_{\a\b,u}=0 .
\ee
With this, from \eqref{const3} with \eqref{Rijk}, \eqref{Ruabg}, \eqref{R1i} and \eqref{H_quad} one obtains 
\be
 H^{(1)}_{,\a}=\frac{Q'(\bar\lambda)}{P'(\bar\lambda)}\left(\bar\lambda W_\a+\frac{1}{n-3}W^{\ \ \ \ \ \b}_{[\a||\b]}\right) , \label{H1_const} 
\ee
such that $\Delta H^{(1)}=\bar\lambda\frac{Q'(\bar\lambda)}{P'(\bar\lambda)}W^{\a}_{\ \ ||\a}$, 
 while \eqref{const4} with \eqref{R11} gives $H^{(0)}$ as a solution of
\beqn
  & & \Delta H^{(0)}-H^{(1)}W^{\a}_{\ \ ||\a}-2W^{\a}H^{(1)}_{,\a}-W^{[\a||\b]}W_{[\a||\b]}-W^{\ \ \ ||\a}_{\a,u} 		+\frac{2\Lambda_0}{n-2}W^\a W_\a \nonumber \\
  & & \qquad =\frac{P''(\bar\lambda)}{(n-3)(n-4)P'(\bar\lambda)}\left(-2W_{[\a||\b]}^{\ \ \ \ \ \a}W^{[\gamma||\b]}_{\ \ \ \ \ \gamma}+W_{[\a||\b]\gamma}W^{[\a||\b]\gamma}\right) . 	\label{H0_const}
\eeqn

The simplest possible solution is obtained for $H^{(1)}=H^{(0)}=W_\a=0$, giving rise to a direct product spacetime $M_2\times\Sigma_{n-2}$ of two spaces of constant curvature (the signs of the two curvatures being those of $\Lambda_0$ and $\bar\lambda$, respectively). This is a Nariai-type geometry already considered in \cite{MaeWilRay11} (and earlier in \cite{Lorenz-Petzold87} in the Gauss-Bonnet case), of (aligned) Ricci and Weyl type D (it becomes conformally flat iff $2\Lambda_0=-(n-2)\bar\lambda$ $\Leftrightarrow Q'(\bar\lambda)=-P'(\bar\lambda)$, which is not possible in Einstein's theory -- cf.~\cite{Ficken39} and the review \cite{OrtPraPra13rev}).

More general solutions with non-zero functions $H^{(1)}$, $H^{(0)}$ or $W_\a$ generically possess also Ricci and Weyl components of negative b.w. and represent gravitational waves propagating in the Nariai-type background. All the curvature scalar invariants, however, are independent of $H^{(1)}$, $H^{(0)}$ and $W_\a$ \cite{ColHerPel10}, so that all such spacetimes are CSI. These metrics cannot generically be \pp waves -- this occurs iff $\Lambda_0=0$ ($\Leftrightarrow Q'(\bar\lambda)=0$), which singles out a class of degenerate theories for which $\bar\lambda$ is (at least) a {\em double} non-zero root of $Q(\lambda)$, corresponding to a doubly degenerate vacuum (not permitted in Einstein's theory).\footnote{An example is given by a quadratic theory with $(n-2)(n-3)c_1^2=4(n-4)(n-5)c_0c_2$, for which $2(n-4)(n-5)\bar\lambda=-c_1/c_2$ (the case $n=5$ gives rise to pure Gauss-Bonnet theory and is special in the sense that $Q(\lambda)=0=Q'(\lambda)$ identically, so that $\bar\lambda$ remains arbitrary -- cf. also \cite{DadPon13}). For $n=6$, this defines the critical point found in \cite{KasSen15}.\label{footn_L=0}}

The above Lovelock solutions do not solve Einstein's theory, generically. However, it can be seen from \eqref{const1}--\eqref{const4} (with \eqref{Riem_const}, \eqref{Hrr}, \eqref{R01}) that such spacetimes are Einstein precisely when $P''(\bar\lambda)=0$ ($\Leftrightarrow 2\Lambda_0=(n-2)(n-3)\bar\lambda$). This condition thus defines a special class of Lovelock theories for which (when restricted to spacetimes admitting a recurrent null vector field with a base space of constant curvature) any solution  is necessarily an Einstein spacetime.
Note that this is not possible in Gauss-Bonnet gravity (and thus requires $n>6$), since in that case $c_2\neq0$ implies $P''\neq0$, independently of $\lambda$.\footnote{An example is given by a cubic theory ($n>6$) with 
$27(n-5)^2(n-6)^2c_0c_3^2+2(n-2)(n-3)(n-4)^2c_2^2=9(n-2)(n-3)(n-5)(n-6)c_1c_2c_3$, for which $3(n-5)(n-6)\bar\lambda=-c_2/c_3$.} 

In the limit $\bar\lambda=0$, one recovers the results of section~\ref{subsubsec_flat_gen}.

\subsubsection{Degenerate Lovelock theories ($P'(\bar\lambda)=0$)}

\label{subsubsec_const_deg}

If $P'(\bar\lambda)=0$ (which gives $Q'(\bar\lambda)=-2\bar\lambda P''(\bar\lambda)$), eqs.~\eqref{const1} and \eqref{const2} imply that also $P(\bar\lambda)=0=P''(\bar\lambda)$, i.e., $\bar\lambda$ is a (at least) {\em triple} root of $P(\lambda)$. This can occur only for degenerate Lovelock theories (with $n\ge7$). In this case, the field equations \eqref{const1}--\eqref{const4} are identically satisfied, and {\em any} metric \eqref{Kundt_gen} with \eqref{Riem_const} is a solution (including the case $\bar\lambda=0$, cf. section~\ref{subsubsec_flat_deg}). For the special case of direct products $M_2\times\Sigma_{n-2}$ (with $M_2$ arbitrary and $\Sigma_{n-2}$ of constant curvature) this was noticed already in \cite{Wheeler86} (see also \cite{MaeWilRay11}). Since $H$ is arbitrary, these metrics are not CSI, in general.

\section{Base space as a direct product of two spaces of constant curvature}

\label{sec_products}

Here we consider the case when the base space is the direct product of two spaces of constant curvature and respective dimensions $n_1$ and $n_2$ (with $n_1+n_2=n-2$). Let us define indices $A,B,C,\ldots$ and $I,J,K,\ldots$ in these two spaces. Well-known properties of direct product spaces \cite{Ficken39} imply that the Riemann tensor of the base space has only non-mixed components, namely
\be
 \hat R^{AB}_{CD}=\lambda_1\delta^{AB}_{CD} , \qquad \hat R^{IJ}_{KL}=\lambda_2\delta^{IJ}_{KL} , 
	\label{riem_prod}
\ee
where $\lambda_1$ and $\lambda_2$ are constants (not both vanishing). Defining 
\beqn
 & & P_1(\lambda_1,\lambda_2)\equiv \sum_{q=0}^{[n_2/2]}\ \ \sum_{k=q}^{q+[(n_1+1)/2]}c_k\binom{k}{q}\frac{\lambda_1^p\lambda_2^{q}}{(n_1-2p+1)!(n_2-2q)!} , \qquad p\equiv k-q  , \\
 & & P_2(\lambda_1,\lambda_2)\equiv \sum_{q=0}^{[n_1/2]}\ \ \sum_{k=q}^{q+[(n_2+1)/2]}c_k\binom{k}{q}\frac{\lambda_1^{q}\lambda_1^p}{(n_1-2q)!(n_2-2p+1)!} , \\
 & & Q(\lambda_1,\lambda_2)\equiv (n_1+1)P_1(\lambda_1,\lambda_2)-2\lambda_1P_{1,\lambda_1}(\lambda_1,\lambda_2)=(n_2+1)P_2(\lambda_1,\lambda_2)-2\lambda_2P_{2,\lambda_2}(\lambda_1,\lambda_2) , \label{Q_prod}
\eeqn
thanks to~\eqref{riem_prod} the field equations~\eqref{Lov01} and \eqref{Lovij} take the form (the latter splitting into $A$- and $I$-components)
\beqn
 & & Q(\lambda_1,\lambda_2)=0 , \label{Lov01_prod} \\
 & & n_1Q(\lambda_1,\lambda_2)-2\lambda_1 Q_{,\lambda_1}(\lambda_1,\lambda_2)+2H_{,rr}P_{1,\lambda_1}(\lambda_1,\lambda_2)=0 , \label{LovAB} \\
 & & n_2Q(\lambda_1,\lambda_2)-2\lambda_2 Q_{,\lambda_2}(\lambda_1,\lambda_2)+2H_{,rr}P_{2,\lambda_2}(\lambda_1,\lambda_2)=0 . \label{LovIJ} 
\eeqn

From now on, in the line-element~\eqref{Kundt_gen} we make the simplifying assumption
\be
 g_{\a\b,u}=0 , \qquad W_\a=0 ,
\ee
so that ${\cal R}_{ijk}=0$. The remaining field equations~\eqref{Lov1i} and \eqref{Lov11} can now be written as
\beqn
 & & P_{1,\lambda_1}{\cal R}_A=0 , \qquad P_{2,\lambda_2}{\cal R}_I=0 , \label{prod_-1} \\
 & & n_2P_{1,\lambda_1}{\cal R}^A_A+n_1P_{2,\lambda_2}{\cal R}^I_I=0 . \label{prod_-2}
\eeqn

The discussion now parallels that of section~\ref{subsec_const_base}. Generically, eq.~\eqref{LovAB} (or \eqref{LovIJ}) implies that $H$ is of the form \eqref{H_quad} and fixes the constant $\Lambda_0$. 
Eqs.~\eqref{prod_-1} and \eqref{prod_-2} then read
\beqn
 & & P_{1,\lambda_1}H^{(1)}_{,A}=0 , \qquad P_{2,\lambda_2}H^{(1)}_{,I}=0 , \label{prod_-1_gen} \\
 & & n_2P_{1,\lambda_1}\Delta_1 H^{(0)}+n_1P_{2,\lambda_2}\Delta_2 H^{(0)}=0 , \label{prod_-2_gen}
\eeqn
where $\Delta_1$ and $\Delta_2$ are the Laplace operators in the geometries of the two factor-subspaces of the base metric. We observe that if both $P_{1,\lambda_1}\neq0\neq P_{2,\lambda_2}$, then \eqref{prod_-1} implies that $H^{(1)}=H^{(1)}(u)$, which is thus removable. An explicit example in six dimensions (with $n_1=2=n_2$) was mentioned in footnote~\ref{footn_ex_5D_6D}.

For special theories such that $P_{1,\lambda_1}=0=P_{2,\lambda_2}$, eqs.~\eqref{Q_prod}--\eqref{LovIJ} further require $P_1(\lambda_1,\lambda_2)=0=P_2(\lambda_1,\lambda_2)$ and $\lambda_1 Q_{,\lambda_1}=0=\lambda_2 Q_{,\lambda_2}$, but do not determine $H_{,rr}$. Furthermore, eqs.~\eqref{prod_-1} and \eqref{prod_-2} are identically satisfied and the function $H(u,r,x)$ is arbitrary -- this is a special instance of the case discussed in section~\ref{subsec_Hspecial}.

\section*{Acknowledgments}

I am grateful to Sourya Ray for many stimulating discussions and for helpful comments. 
This work has been supported by research plan RVO: 67985840 and research grant GA\v CR~13-10042S.   
The author's stay at Instituto de Ciencias F\'{\i}sicas y Matem\'aticas, Universidad Austral de Chile has been supported by CONICYT PAI ATRACCI{\'O}N DE CAPITAL HUMANO AVANZADO DEL EXTRANJERO Folio 80150028.

\renewcommand{\thesection}{\Alph{section}}
\setcounter{section}{0}

\renewcommand{\theequation}{{\thesection}\arabic{equation}}

\section{A few general results on Lovelock vacua}
\setcounter{equation}{0}

\label{app_general}

In the following, we discuss a few results on exact vacuum solution of Lovelock gravity~\eqref{fieldeqns}. These are {\em not} restricted to spacetimes admitting a recurrent null vector field and thus apply in a more general context.

\subsection{Riemann type N}

An observation of \cite{ReaTanBen14} can be slightly rephrased as (see also \cite{PraPra08} for the special case of Gauss-Bonnet gravity) 
\begin{proposition}[\cite{ReaTanBen14}]
\label{prop_RiemN}
 A spacetime of Riemann type N is a vacuum solution of Lovelock gravity iff $c_0=0=c_1R_{11}$.
\end{proposition}
\begin{proof} 
 A Riemann tensor of type N possesses only components of b.w. $-2$ (in particular, $R=0$, and $R_{11}$ is the only non-zero Ricci component). This means that only the cosmological and the Einstein terms survive in the vacuum Lovelock equations (i.e., terms quadratic and of higher order in the curvature tensor vanish identically), from which the result follows immediately. 
\end{proof} 

The cosmological constant is thus necessarily zero and, generically (i.e., $c_1\neq0$), such solutions are Ricci-flat and of Weyl type N (see, e.g., \cite{OrtPraPra13rev} and references therein for some examples); conversely, any Ricci-flat spacetime of Weyl type N solves Lovelock gravity with $c_0=0$ (this was noticed in \cite{BouDes85,GibRub86,GleDot05} in special cases). For special Lovelock theories with $c_0=0=c_1$ (such as pure Gauss-Bonnet gravity), any metric of Riemann type N is a solution. 

\subsection{Riemann type III}

Similar considerations lead to the following generalization for type III (for which we omit a similar, straightforward proof, based on~\eqref{fieldeqns} with the observation that $G^a_c$ can now contain at most quadratic terms) -- the notation~\eqref{Rcomp_1} for Riemann components of b.w.~$-1$ is used.
\begin{proposition}
\label{prop_RiemIII}
 A spacetime of Riemann type III is a vacuum solution of Lovelock gravity iff $c_0=0$ and 
 \beqn
  & & c_1R_{1i}=0 , \label{RIII_1} \\
  & & c_1R_{11}+2c_2(-2{\cal R}^{ij}_i{\cal R}^l_{lj}+{\cal R}^{jl}_i{\cal R}^i_{jl})=0 . \label{RIII_2} 
 \eeqn
\end{proposition}

Again, the cosmological constant is necessarily zero. Generically, the Ricci type must be N (which means that ${\cal R}_{ijk}=C_{1ijk}$ in \eqref{RIII_2}). For Ricci-flat metrics, eq.~\eqref{RIII_2} reduces to a condition found in \cite{MalPra11prd} for Einstein spacetimes in quadratic gravity. For special theories with $c_0=0=c_1$, the only surviving equation takes the form $c_2(-2{\cal R}^{ij}_i{\cal R}^l_{lj}+{\cal R}^{jl}_i{\cal R}^i_{jl})=0$, while $R_{1i}$ is unrestricted. For even more special theories with $c_0=c_1=c_2=0$ ($n\ge7$), any metric of Riemann type III is a solution.

From \cite{PodOrt06} with Propositions~\ref{prop_RiemN} and \ref{prop_RiemIII} it follows, for example, that Robinson-Trautman vacua of Riemann type III/N do not exist in theories with $c_1\neq0$.

Recall that metrics of Riemann type III/N include all the (non-flat) VSI spacetimes \cite{Coleyetal04vsi}.

\subsection{Weyl and traceless-Ricci type III (aligned)}

\label{app_RicciIII}

A further generalization enables one to include a cosmological constant as follows (as noticed in section~\ref{sec_recurrent}, spacetimes admitting a recurrent null vector field do not fall in this class, except when $\lambda=0$).
\begin{proposition}
\label{prop_Riem_align}
 A spacetime of (aligned) Weyl and traceless-Ricci type III (i.e., $R^{ab}_{cd}=\lambda\delta^{ab}_{cd}+(\mbox{b.w.}<0)$)
 is a vacuum solution of Lovelock gravity iff 
 \beqn 
  & & P(\lambda)=0 , \label{align_1} \\ 
  & & P'(\lambda)R_{1i}=0 , \label{align_2} \\
  & & (n-3)(n-4)P'(\lambda)R_{11}+P''(\lambda)(-2{\cal R}^{ij}_i{\cal R}^l_{lj}+{\cal R}^{jl}_i{\cal R}^i_{jl})=0 , \label{align_3} 
 \eeqn
where 
\be
 P(\lambda)\equiv \sum_{k=0}^{[(n-1)/2]} c_k\frac{\lambda^k}{(n-2k-1)!} ,
\ee
and $P'$, $P''$ are the derivatives of $P$ w.r.t. $\lambda$.
\end{proposition}
\begin{proof} 

The above result follows by plugging $R^{ab}_{cd}=\lambda\delta^{ab}_{cd}+(\mbox{b.w.}<0)$ into \eqref{fieldeqns} (with the notation \eqref{Rcomp_1}), after simple combinatorics and recalling~\eqref{ident_delta}. (In passing, we note that \eqref{align_1} implies that $\lambda$ is a constant, so that this need not be assumed.)

\end{proof} 

The polynomial eq.~\eqref{align_1} generically possesses $[(n-1)/2]$ solutions, thus fixing $\lambda$ in terms of the $c_k$  (we assume that at least one, say $\bar\lambda$, is real and non-zero), exactly as in the well-known case of constant curvature spacetimes \cite{BouDes85,Wheeler86}. 
For a {\em non-degenerate} Lovelock theory, eq.~\eqref{align_2} implies $R_{1i}=0$, i.e., the traceless-Ricci type must be N. If the Weyl type is assumed to be also N, then \eqref{align_3} implies $R_{11}=0$, i.e., the spacetime must be Einstein (in agreement with results obtained in \cite{GibRub86,GleDot05} in special cases). Assuming, instead, the spacetime to be Einstein (without restricting the Weyl type), eq.~\eqref{align_3} reduces to 
\be
 P''(\bar\lambda)(-2{\cal R}^{ij}_i{\cal R}^l_{lj}+{\cal R}^{jl}_i{\cal R}^i_{jl})=0 \qquad \mbox{(for Einstein spacetimes)} . 
 \label{align_einst}
\ee
As mentioned above, the same condition was found in \cite{MalPra11prd} for solutions to quadratic gravity. Einstein spacetimes of genuine Weyl type III satisfying \eqref{align_einst} (and thus Lovelock's theory) are known, cf., e.g., \cite{MalPra11prd,OrtPraPra10}.

On the other hand, if the theory is degenerate and $\bar\lambda$ is taken to be (at least) a {\em double} root of $P(\lambda)$, eq.~\eqref{align_2} becomes an identity and \eqref{align_3} reduces again to \eqref{align_einst} -- this may give rise to solutions containing arbitrary functions, cf. one such example in \cite{GleDot05}. If $\bar\lambda$ is (at least) a {\em triple} root ($n\ge7$), then both \eqref{align_2} and \eqref{align_3} are identically satisfied, and {\em any} spacetime of (aligned) Weyl and traceless-Ricci type III is a solution.

%
%
%
%

\providecommand{\href}[2]{#2}\begingroup\raggedright\endgroup

\end{document}